\documentclass[aps,pra,twocolumn,superscriptaddress,nofootinbib]{revtex4-2}
\usepackage{amsmath,amssymb,amsthm}
\usepackage{graphicx}
\usepackage{booktabs}
\usepackage{hyperref}

\newtheorem{theorem}{Theorem}
\newtheorem{lemma}{Lemma}

\begin{document}

\title{Exact minimum measurement dependence for faithful local deterministic models of multipartite GHZ--Mermin correlations}

\author{Aaron Alai}
\affiliation{Baltimore, Maryland, USA}

\date{October 20, 2026 (v3); first version July 30, 2026}

\begin{abstract}
A local deterministic model can reproduce Bell-inequality violations only if the hidden state depends on the measurement settings. Hall found the minimum such dependence for the singlet state [Phys.\ Rev.\ Lett.\ \textbf{105}, 250404 (2010)] and posed the corresponding problem for the Mermin state [Phys.\ Rev.\ A \textbf{84}, 022102 (2011)]. This paper solves it, for the $n$-party GHZ--Mermin scenario at every $n \ge 3$, with the marginals reproduced as well as the correlators. The minimum measurement-dependence fraction is $F_{\min}(n) = R/[2(R+1)]$, with $R = 2^{\lfloor (n-1)/2\rfloor}$ the Mermin violation ratio: $1/3, 2/5, 4/9, 8/17, \ldots$, always below $1/2$. The lower bound follows from Mermin's inequality by a general overlap bound; the upper bound is an explicit model, uniform on the hidden-variable classes that violate the fewest constraints, with closed-form counts; a reduction theorem shows that for Mermin scenarios the marginals impose no additional dependence. The Mermin values are also certified by integer-arithmetic enumeration through $n = 22$. The formula was stated as a conjecture in the first version of this paper and was first proved, by a different route, by Wang, Yao and Chen; the proof here was obtained independently and the two agree. Every minimal contradiction of size $k$ among stabilizer constraints has correlator-only minimum $1/(k-1)$. With the marginals included the minimum is the same on every core of the $C_5$ ring cluster, on the size-4 and size-5 cores of $C_4$ and on every four-core of $C_6$, but larger on the size-6 and size-7 cores of $C_4$ and on $112$ of the $1{,}920$ five-cores of $C_6$, with exact certificates on both sides in every case.
\end{abstract}

\maketitle

\section{Introduction}

Bell inequalities follow from three assumptions about a hidden-variable model: locality and outcome determinism (jointly replaced, in the stochastic formulation, by Bell's local causality), and measurement independence---the statistical independence of the setting choices from the hidden state of the measured system \cite{Bell1964}. Experiments violate the inequalities, so at least one assumption fails. Most responses abandon locality or determinism. This paper concerns the third option: models in which the hidden state depends on the settings by a bounded, quantified amount.

The quantitative study of such models is smaller but well developed. Brans gave an existence proof \cite{Brans1988}. Hall determined the exact minimum dependence for the singlet state \cite{Hall2010}, later refined per observer \cite{Friedman2019}, compared between causal and retrocausal models with Branciard \cite{HallBranciard2020}, and extended to Kochen--Specker scenarios \cite{Hall2011}. Barrett and Gisin bounded the dependence in mutual information \cite{BarrettGisin2011}; Wood and Spekkens analyzed it from the standpoint of causal discovery \cite{WoodSpekkens2015}; Waegell and McQueen distinguished conspiratorial from nomic dependence \cite{WaegellMcQueen}; Sen and Valentini constructed explicit models \cite{SenValentini2020}.

Two of Hall's results frame the present work. For the singlet state, a local deterministic model must have measurement-dependence fraction $F \ge (\sqrt{2}-1)/3 \approx 13.81\%$ in the measure defined below, with equality attainable (the mutual-information bound of Ref.~\cite{BarrettGisin2011} is a distinct result in a distinct measure). For the tripartite Mermin scenario, Hall \cite{Hall2011} constructed a model reproducing the four perfect correlators at $F = 1/3$, with no constraint on the marginals, and posed the construction of ``a similar optimal model for Mermin's state'' as an open problem. Related tripartite results bound the dependence sufficient to reach relaxed inequality values \cite{Hossain2021,Roy2013}, and recent work bounds task-level dependence for nonlocal games \cite{Pal2026}. None computes the minimum total-variation dependence (Hall's measure) for full faithful reproduction of the statistics---every correlator and every marginal---which is the object of this paper. After the first version of this paper, which stated the closed form below as a conjecture, Wang, Yao and Chen proved it for every $n$ by a different route and extended the problem to dephased GHZ states and to the full Pauli $X/Y$ setting set \cite{Wang2026}; Sec.~\ref{sec:staircase} records the relation between the two proofs.

This paper solves the faithful problem exactly at every size. The contributions are: (i) a reduction theorem showing that the marginal (faithfulness) constraints are free, so that Hall's correlator-only tripartite value is also the faithful value $F(3) = 1/3$, and that the optimization collapses to a tractable class space; (ii) the exact minimum $F_{\min}(n) = R/[2(R+1)]$ for every $n \ge 3$, with $R = 2^{\lfloor(n-1)/2\rfloor}$ the Mermin violation ratio \cite{Mermin1990}---discovered by linear programming through $n = 13$, certified in exact integer arithmetic through $n = 22$, and proved in general by a combinatorial lower bound met by an explicit construction whose counts have a closed form (the first proof of the formula, by a different route, is due to Wang, Yao and Chen \cite{Wang2026}); (iii) hand proofs at $n = 3, 4$ that exhibit the structure; (iv) a universal ceiling $F \le 1/2$, so that faithful reproduction never requires total measurement dependence at any size; and (v) an optimal family of hidden-variable densities, uniform on the ground states of the classicalized stabilizer Hamiltonian of the prepared state---proved optimal for the Mermin sets and, for the correlators alone, for every minimal contradiction of stabilizer constraints; checked exactly, marginals included, on every minimal contradiction of the $C_5$ ring cluster, on the size-4 and size-5 ones of $C_4$, and on $280$ of the $360$ four-cores and $1{,}616$ of the $1{,}920$ five-cores of $C_6$; and shown by exact certificates to fail on the larger cores of $C_4$ and on $112$ five-cores of $C_6$, where the faithful minimum exceeds $1/(k-1)$.

Section~\ref{sec:scenario} sets up the scenario, the measure of dependence, and the linear program. Section~\ref{sec:reduction} states the two reductions that make the program tractable. Section~\ref{sec:lowerbound} states the lower bound and Sec.~\ref{sec:construction} gives the construction that meets it; Sec.~\ref{sec:exact} reports the certified values. Sections~\ref{sec:hand}--\ref{sec:staircase} give hand proofs at $n = 3, 4$, the ceiling, and the closed form at every $n$. Section~\ref{sec:structure} examines the structure of the optimal densities and tests it on cluster states. Complete proofs of Theorems~\ref{thm:reduction}, \ref{thm:symmetry}, \ref{thm:lowerbound}, and \ref{thm:n4}, and the Fourier-analytic lemmas behind Theorems~\ref{thm:allnlower} and \ref{thm:staircase}, are given in Appendices~\ref{app:reduction}--\ref{app:fourier}.

\section{Scenario, measure, and linear program}\label{sec:scenario}

For $n \ge 3$ parties sharing the $n$-partite GHZ state $(|0\rangle^{\otimes n} + |1\rangle^{\otimes n})/\sqrt{2}$, each party measures Pauli $X$ or $Y$. The Mermin settings are strings $s \in \{X,Y\}^n$ with an even number of $Y$'s; there are $2^{n-1}$ of them. Quantum mechanics predicts with certainty
\begin{equation}
E(s) = \begin{cases} +1, & \#Y(s) \equiv 0 \pmod 4,\\ -1, & \#Y(s) \equiv 2 \pmod 4,\end{cases}
\label{eq:targets}
\end{equation}
with every proper-subset correlator vanishing. Throughout, $R := 2^{\lfloor (n-1)/2\rfloor}$ denotes the ratio of the quantum value, $2^{n-1}$, of the Mermin expression $\sum_s E(s)\,\langle s\rangle$ to its local bound \cite{Mermin1990}; equivalently, a deterministic assignment can satisfy at most a fraction $s = (R+1)/(2R)$ of the $2^{n-1}$ conditions \eqref{eq:targets}.

A local deterministic model assigns each hidden state $\lambda$ pre-set answers $(a_i, b_i) \in \{\pm 1\}^2$ per party---$4^n$ deterministic strategies---and, for each settings string $s$, a probability density $\rho_s(\lambda)$. Determinism costs nothing here: local randomness can be absorbed into $\lambda$ as a settings-independent factor, which changes no distance between the $\rho_s$ (as in Ref.~\cite{Hall2010}). Following Hall \cite{Hall2010}, the degree of measurement dependence is
\begin{equation}
M := \max_{s,s'} \sum_\lambda \left| \rho_s(\lambda) - \rho_{s'}(\lambda) \right|,
\end{equation}
the largest $L_1$ distance between the densities of two settings. It is convenient to work with the normalized quantity. For a pair of settings let $F_{ss'} := \tfrac{1}{2}\sum_\lambda |\rho_s(\lambda) - \rho_{s'}(\lambda)| \in [0,1]$ be the total-variation distance between $\rho_s$ and $\rho_{s'}$, and call
\begin{equation}
F := M/2 = \max_{s,s'} F_{ss'}
\end{equation}
the \emph{measurement-dependence fraction}, so that $M = 2F$ is Hall's un-normalized form: $F = 0$ is measurement independence and $F = 1$ is disjoint support for some pair of settings. (Ref.~\cite{Hall2011} uses the symbol $F$ for the retained fraction $1 - M/2$; here $F$ always denotes $M/2$.) A model is \emph{faithful} if it reproduces every full correlator $E(s)$ and every vanishing proper-subset correlator. (In the stabilizer scenarios of Sec.~\ref{sec:structure} each party has a pre-set outcome for each of $X$, $Y$, $Z$---$8^n$ strategies; in the setting of a stabilizer $g$ the parties in its support measure the indicated Pauli and the others are not measured; and faithful means that for every subset $T$ of the support the product of outcomes over $T$ has expectation $+1$ if the unsigned Pauli string $g|_T$ is in the stabilizer group with sign $+$, $-1$ if with sign $-$, and $0$ if it is not in the group. For the GHZ--Mermin scenario all proper-subset values vanish. Outcome bits for Paulis a party never measures in a scenario are inert---dropping them changes no constraint and no distance---so linear programs and certificates are computed on the strategy space collapsed to the measured Paulis.) The marginals are \emph{free} in a scenario if its faithful minimum equals its correlator-only minimum, the minimum over models required to reproduce the full correlators only. The values reported below are minima of $F$ over all faithful models for the stated finite setting sets; richer scenarios containing them can only raise the values.

\emph{The linear program.} With variables $\rho_s(\lambda) \ge 0$, the constraints are linear: normalization, full correlators equal to $E(s)$, and all proper-subset correlators equal to zero. The objective---minimize the maximum pairwise $L_1$ distance---is encoded with standard auxiliary variables. The programs were solved numerically in double precision. Every Mermin value reported below is then \emph{certified exactly}, with no dependence on solver tolerance, by showing that the lower bound of Theorem~\ref{thm:lowerbound} and the construction of Sec.~\ref{sec:construction} coincide in integer arithmetic (Secs.~\ref{sec:exact} and \ref{sec:methods}).

\section{Two reductions}\label{sec:reduction}

\begin{theorem}[Fiber reduction]\label{thm:reduction}
Write $a_i = (-1)^{\alpha_i}$, $b_i = (-1)^{\beta_i}$, $P := \bigoplus_i \alpha_i$, $\gamma_i := \alpha_i \oplus \beta_i$, and call $c(\lambda) := (P, \gamma) \in \mathrm{GF}(2)^{n+1}$ the class of $\lambda$. The faithful linear program and the program with variables $\tilde\rho_s(c) \ge 0$ on the $2^{n+1}$ classes, constrained only by normalization and the full-correlator conditions $\sum_c \tilde\rho_s(c)\,(-1)^{P \oplus \bigoplus_{i \in S}\gamma_i} = E(s)$, have the same optimal value.
\end{theorem}

The idea is that every Mermin correlator is the parity $P \oplus \bigoplus_{i \in S}\gamma_i$ for the relevant $Y$-subset $S$, so all constraint functionals depend on $\lambda$ only through its class; each class is a fiber of $2^{n-1}$ strategies, and the fiber-uniform lift annihilates every proper-subset correlator while a group-averaging argument shows the lift does not increase the objective. The full proof is given in Appendix~\ref{app:reduction}.

Theorem~\ref{thm:reduction} has two consequences. Conceptually, the faithfulness constraints are \emph{free}: the correlator-only optimum and the full faithful optimum provably coincide, so the faithful minimum at $n = 3$ equals Hall's correlator-only value $1/3$ by Theorem~\ref{thm:reduction}, not by coincidence. Practically, $n = 7$ collapses from $16{,}384$ strategies to $256$ classes.

\begin{theorem}[Lossless symmetrization]\label{thm:symmetry}
The scenario is invariant under party permutations; the feasible set is convex and the objective is convex and permutation-invariant; hence group-averaging any optimum yields a covariant optimum, and restricting to $S_n$-covariant densities---functions of $(|S|, P, |\gamma \cap S|, |\gamma \cap \bar{S}|)$ only, with $S$ the $Y$-subset of the setting---does not change the optimal value. The number of density variables drops from $O(4^n)$ to $O(n^3)$, and the auxiliary variables of the $L_1$ encoding reduce likewise to polynomially many.
\end{theorem}

The full proof, including the completeness of the listed invariants, is given in Appendix~\ref{app:symmetry}. With Theorem~\ref{thm:symmetry}, the $n = 9$ program falls from roughly $33$ million variables (including the auxiliary variables of the $L_1$ encoding) to $3{,}673$, and the $n = 13$ minimum solves in seconds on commodity hardware.

\section{A lower bound}\label{sec:lowerbound}

\begin{theorem}[Overlap bound]\label{thm:lowerbound}
Let $S'$ be a set of $|S'| \ge 2$ settings of a scenario whose full correlators have deterministic targets ($\pm 1$), and let $m(S')$ be the maximum, over strategies, of the number of settings in $S'$ whose constraint the strategy satisfies. Then that scenario's minimum satisfies
\begin{equation}
F_{\min} \;\ge\; \frac{|S'| - m(S')}{|S'| - 1},
\label{eq:overlapbound}
\end{equation}
and hence $F_{\min} \ge B := \max_{S'} (|S'| - m(S'))/(|S'| - 1).$
\end{theorem}

The proof rests on a support restriction (deterministic targets force each density onto its satisfying set), a pointwise application of Chebyshev's sum inequality, and a pigeonhole step; it is given in full in Appendix~\ref{app:overlap}. The bound applies to any scenario whose constraints have deterministic targets, not only to Mermin scenarios (the proof uses nothing else). Two terms used below: a \emph{contradiction} is a set of constraints whose $\mathrm{GF}(2)$ functionals sum to zero while their targets sum to one, so that no strategy (no class, in a Mermin scenario) satisfies all of them; a \emph{core} is a minimal contradiction. On the Mermin sets the bound has a closed value at every size:

\begin{theorem}[Lower bound at every $n$]\label{thm:allnlower}
$F_{\min}(n) \ge R/[2(R+1)]$ for every $n \ge 3$.
\end{theorem}

\begin{proof}
For odd $n$, Lemma~\ref{lem:mermsum} (Appendix~\ref{app:fourier}) gives the maximum satisfiable count on the full set of $N = 2^{n-1} = R^2$ settings as $m = (N + R)/2$, so Theorem~\ref{thm:lowerbound} with $S'$ the full set yields $(N - m)/(N - 1) = (R^2 - R)/(2(R^2 - 1)) = R/[2(R+1)]$. For even $n$, take $S'$ to be the $2^{n-2}$ settings in which party $n$ measures $X$. Their constraint functionals $P \oplus \bigoplus_{i \in S} \gamma_i$ with $S \subseteq [n-1]$ do not involve $\gamma_n$, and their targets depend on $|S|$ only, so as functions of $(P, \gamma_1, \ldots, \gamma_{n-1})$ they are exactly the constraints of the Mermin$_{n-1}$ scenario; hence $m(S')$ is the odd-size count just computed at $n - 1$, and $R(n-1) = R(n)$.
\end{proof}

Theorem~\ref{thm:allnlower} is Mermin's inequality \cite{Mermin1990} fed through Theorem~\ref{thm:lowerbound}; it is the lower-bound half of Theorem~\ref{thm:staircase}.

\section{A construction}\label{sec:construction}

Define the frustration Hamiltonian $H(\lambda)$ as the number of Mermin constraints violated by $\lambda$'s class, and the \emph{sector} of a setting $s$ as the set $A_s$ of classes satisfying its constraint. No class satisfies every constraint, so $\min H \ge 1$; this is the GHZ contradiction. The \emph{ground-state construction} takes $\rho_s$ uniform over the ground states of $H$ (the classes of minimum violation count) within the sector of $s$. Its value of $F$ is a ratio of integer counts, so it can be evaluated exactly. Lemma~\ref{lem:mermsum} (Appendix~\ref{app:fourier}) gives the ground set in closed form---$2^n$ classes for odd $n$, $2^{n-1}$ for even $n$---and Lemma~\ref{lem:counts} counts it sector by sector: every sector contains $g = R(R+1)$ ground states, and two distinct sectors share $R(R+2)/2$ of them (except that at even $n$ complementary settings receive identical densities). Hence the construction is feasible at every $n$ and $F_{ss'} = 1 - (R+2)/[2(R+1)] = R/[2(R+1)]$ for every pair of distinct, non-complementary settings (complementary pairs at even $n$ have $F_{ss'} = 0$), which is the value of Theorem~\ref{thm:allnlower}. That is the upper-bound half of Theorem~\ref{thm:staircase}; the exact certification of Sec.~\ref{sec:exact} enumerates the same counts, one setting pair per $S_n$-orbit (Theorem~\ref{thm:symmetry}), through $n = 22$.

\section{Exact values}\label{sec:exact}

Table~\ref{tab:floors} collects the exact minima. The Mermin values are those of Theorem~\ref{thm:staircase}; each is also certified by enumeration, in integer arithmetic and without the closed forms of Lemmas~\ref{lem:mermsum}--\ref{lem:counts}, by matching bounds: the lower bound of Theorem~\ref{thm:lowerbound}, evaluated from exactly enumerated satisfiability counts (on the full set for odd $n$ and on the embedded Mermin$_{n-1}$ sub-scenario for even $n$), coincides with the exact rational value of the construction of Sec.~\ref{sec:construction}, obtained by enumerating the ground set and the sector overlaps per $S_n$-orbit; the enumeration also asserts the counts of Lemma~\ref{lem:counts} at each $n$. The certification runs through $n = 22$; the table lists $n \le 13$ and the odd sizes beyond, each even size $14 \le n \le 22$ repeating the value at $n - 1$. The linear program, run in double precision and reproduced on independent hardware with independent solver algorithms (dual simplex and interior point), served for discovery and as a cross-check; the certification script is part of the verification suite (Sec.~\ref{sec:methods}).

\begin{table}[t]
\caption{\label{tab:floors}Minimum measurement-dependence fraction $F_{\min}$ for faithful reproduction of the GHZ--Mermin correlations. $s$ is the maximum classical satisfiability of the constraint set. Every Mermin row is Theorem~\ref{thm:staircase}, certified independently by enumeration in integer arithmetic (Sec.~\ref{sec:exact}); the CHSH row is a numerical (LP) reproduction of Hall's analytic value. Cluster-state cores are in Table~\ref{tab:cores}.}
\begin{ruledtabular}
\small\setlength{\tabcolsep}{3pt}
\begin{tabular}{lccc}
Scenario & $F_{\min}$ & $s$ & $F_{\min}s$ \\
\colrule
CHSH at Tsirelson ($n=2$) & $(\sqrt{2}-1)/3$ & --- & --- \\
Mermin $n=3$ & $1/3$ & $3/4$ & $1/4$ \\
Mermin $n=4$ & $1/3$ & $3/4$ & $1/4$ \\
Mermin $n=5$ & $2/5$ & $5/8$ & $1/4$ \\
Mermin $n=6$ & $2/5$ & $5/8$ & $1/4$ \\
Mermin $n=7$ & $4/9$ & $9/16$ & $1/4$ \\
Mermin $n=8$ & $4/9$ & $9/16$ & $1/4$ \\
Mermin $n=9$ & $8/17$ & $17/32$ & $1/4$ \\
Mermin $n=10$ & $8/17$ & $17/32$ & $1/4$ \\
Mermin $n=11$ & $16/33$ & $33/64$ & $1/4$ \\
Mermin $n=12$ & $16/33$ & $33/64$ & $1/4$ \\
Mermin $n=13$ & $32/65$ & $65/128$ & $1/4$ \\
Mermin $n=15$ & $64/129$ & $129/256$ & $1/4$ \\
Mermin $n=17$ & $128/257$ & $257/512$ & $1/4$ \\
Mermin $n=19$ & $256/513$ & $513/1024$ & $1/4$ \\
Mermin $n=21$ & $512/1025$ & $1025/2048$ & $1/4$ \\
\end{tabular}
\end{ruledtabular}
\end{table}

The bipartite row is a consistency check: the four-setting CHSH program at the Tsirelson angles reproduces Hall's continuum minimum $(\sqrt{2}-1)/3$ \cite{Hall2010} from a finite LP. Cluster-state cores are in Table~\ref{tab:cores} (Sec.~\ref{sec:structure}).

\section{Hand proofs at $n = 3$ and $n = 4$}\label{sec:hand}

The two smallest cases of Theorem~\ref{thm:staircase} can be proved by hand; the proofs exhibit the three-of-four structure that the general counting of Lemma~\ref{lem:counts} reproduces at every size.

\begin{theorem}[$n=3$ minimum]\label{thm:n3}
$F_{\min}(3) = 1/3$.
\end{theorem}

\begin{proof}
The four Mermin parities multiply to $+1$ identically while the four targets multiply to $-1$; hence every deterministic strategy satisfies exactly three or exactly one of the four conditions. \emph{Upper bound:} let $C_i$ denote the set of strategies satisfying all conditions but the $i$th; placing weight $1/3$ on each of the three admissible sets per setting gives $F_{ss'} = 1/3$ for every pair, and all marginals vanish by the fiber lift of Theorem~\ref{thm:reduction}. (After fiber-uniformization this construction coincides with the model of Ref.~\cite{Hall2011}, Table~I, for the perfect correlators; the lift is what supplies faithfulness at no cost.) \emph{Lower bound:} each $\lambda$ lies in the support of at most three of the four densities; summing pairwise overlaps over the six pairs bounds the total overlap by $4$, so some pair overlaps at most $2/3$, i.e., $F_{ss'} \ge 1/3$ for some pair.
\end{proof}

\begin{theorem}[$n=4$ minimum]\label{thm:n4}
$F_{\min}(4) = 1/3$.
\end{theorem}

\begin{proof}[Proof sketch]
In outline: GF(2) relations among the eight constraint parities with odd target sums forbid joint satisfaction; optimal classes satisfy exactly six of eight conditions, with the two missed settings always forming a complementary pair, so the missed pairs partition the settings into four complementary pairs; settings across pairs then reproduce the three-of-four combinatorics of Theorem~\ref{thm:n3}. The full proof---including the catalog of minimal contradictions, the classification of the eight ground-state classes, the explicit optimal construction, and the matching lower bound---is given in Appendix~\ref{app:n4}.
\end{proof}

\section{The universal ceiling}\label{sec:ceiling}

\begin{theorem}[Ceiling]\label{thm:ceiling}
$F_{\min}(n) \le 1/2$ for every $n \ge 3$.
\end{theorem}

\begin{proof}
Work at class level (Theorem~\ref{thm:reduction}). Each setting's satisfying set of classes is an affine coset of a hyperplane in $\mathrm{GF}(2)^{n+1}$, of density $1/2$. Distinct settings $S \ne S'$ have distinct linear functionals (their difference is $\bigoplus_{i \in S \triangle S'} \gamma_i \not\equiv 0$), so any two satisfying sets intersect in density exactly $1/4$. Let each $\rho_s$ be uniform on its satisfying set of $2^n$ classes, with value $2^{-n}$ there. The overlap of any pair is then $\sum_c \min(\rho_s, \rho_{s'}) = 2^{n-1} \cdot 2^{-n} = 1/2$, so $F_{ss'} = 1 - 1/2 = 1/2$ for every pair, and marginals vanish by the fiber lift.
\end{proof}

A faithful model with $F = 1/2$ therefore exists at every size: the statistics never \emph{require} the limit $F \to 1$ of complete settings dependence. The geometric content is elementary: satisfying one parity constraint confines the hidden state to half the space, and two such halves, cut by distinct linear functionals, always meet in exactly a quarter.

\section{The closed form}\label{sec:staircase}

\begin{theorem}[Staircase]\label{thm:staircase}
For every $n \ge 3$, with $R := 2^{\lfloor (n-1)/2 \rfloor}$ the Mermin violation ratio and $s := (R+1)/(2R)$ the classical satisfiability,
\begin{equation}
F_{\min}(n) \;=\; \frac{R}{2(R+1)} \;=\; \frac{1}{4\,s}, \qquad F_{\min}\cdot s = \frac{1}{4}.
\end{equation}
\end{theorem}

\begin{proof}
Theorem~\ref{thm:allnlower} gives $F_{\min}(n) \ge R/[2(R+1)]$. The ground-state construction of Sec.~\ref{sec:construction} is a faithful model: each $\rho_s$ is supported on the sector of $s$, so every full correlator is reproduced, and the fiber lift of Theorem~\ref{thm:reduction} annihilates every proper-subset correlator. By Lemma~\ref{lem:counts} every sector contains $R(R+1)$ ground states and two distinct sectors share at least $R(R+2)/2$ of them, so $F_{ss'} \le 1 - (R+2)/[2(R+1)] = R/[2(R+1)]$ for every pair. The identity $1/(4s) = R/[2(R+1)]$ is $s = (R+1)/(2R)$.
\end{proof}

The formula is a staircase in $n$, constant on the pairs $(2k+1, 2k+2)$, converging geometrically to the ceiling $1/2$ and never attaining it; Theorem~\ref{thm:ceiling} is thus also a corollary, though its own proof is independent and shorter. The construction of Sec.~\ref{sec:construction} is optimal at every $n$, and at odd $n$ it is \emph{equidistant}: every pair of settings is at the same distance $R/[2(R+1)]$ (at even $n$, complementary settings coincide and all other pairs are at that distance). The exact certification of Sec.~\ref{sec:exact} reproduces the counts of Lemma~\ref{lem:counts} by enumeration at $n = 3$--$22$ ($R$ from $2$ to $1024$).

\emph{Relation to independent work.} The first version of this paper (arXiv v1, July 2026) proved Theorem~\ref{thm:lowerbound}, computed the values through $n = 13$, and stated Theorem~\ref{thm:staircase} as a conjecture. Wang, Yao and Chen then posted a proof of the formula for every $n \ge 3$ \cite{Wang2026}, the first on the public record: their construction is built from quadratic Walsh spectra, their lower bound from a nondegenerate subsystem of $R^2$ settings with a matching uniform-pair dual, and their faithful lift, like Theorem~\ref{thm:reduction}, restores the vanishing proper-subset correlators at no cost. Later versions of their work extend the result to isotropic visibility $v$, where the cost is $\max\{0, (R^2 v - R)/[2(R^2-1)]\}$ and vanishes for $v \le 1/R$, and to the complete $X/Y$ setting set with both coherence components. The proof given above (Lemmas~\ref{lem:mermsum}--\ref{lem:counts}, the closed form of the Fourier sum of $H$) was obtained independently and by a different route; the two proofs agree in the value at every $n$ and in the counts of the construction. Theorem~\ref{prop:core} and the cluster-core results of Sec.~\ref{sec:structure} are not treated in Ref.~\cite{Wang2026}.

\section{Structure of the optimal densities}\label{sec:structure}

\subsection{Fourier spectrum of $H$}\label{sec:spectrum}

Call the \emph{classicalization} of a Pauli operator the result of replacing each local Pauli by the party's pre-set outcome for it. The Mermin constraints are the classicalizations of the GHZ stabilizers under $X_i \mapsto a_i$, $Y_i \mapsto b_i$, and $H$ is the classicalized stabilizer Hamiltonian; its minimum is positive, which is the GHZ contradiction. Writing $\hat{t}_S := E(s) \in \{\pm 1\}$ for the target of the setting with $Y$-subset $S$, Fourier analysis over $\mathrm{GF}(2)^{n+1}$ (Appendix~\ref{app:fourier}) shows that
\begin{equation}
H \;=\; 2^{n-2} \;-\; \tfrac{1}{2}\sum_S \hat{t}_S\, \chi_S,
\label{eq:spectrum}
\end{equation}
with nonzero coefficients only on the $2^{n-1}$ stabilizer characters $\chi_S$, each of magnitude exactly $1/2$, at odd degrees $1, 3, \ldots, 2\lfloor n/2\rfloor + 1$. In particular $H$ has nonzero coefficients at degree $3$ for every $n \ge 3$, so it is not a one- and two-body function of the class variables $(x_0, x_1, \ldots, x_n)$. Here $\chi_S(c) = (-1)^{P \oplus \bigoplus_{i \in S}\gamma_i}$. Appendix~\ref{app:fourier} shows further that $H$ is the classicalization of the stabilizer Hamiltonian $\sum_S \tfrac{1}{2}(1 - \hat{t}_S O_S)$, with $O_S$ the Mermin operator ($X$ on the parties outside $S$, $Y$ on those in $S$), whose unique ground state is the prepared GHZ state. By Theorem~\ref{thm:staircase}, an optimal family of densities is, at every $n$, uniform within each satisfying sector on the ground states of the classicalized Hamiltonian of the prepared state. For the correlators alone, the same holds for every minimal contradiction of stabilizer constraints (Theorem~\ref{prop:core}). With the marginals included it holds on every core of the $C_5$ ring cluster, on the size-4 and size-5 cores of $C_4$, and on $280$ of the $360$ four-cores and $1{,}616$ of the $1{,}920$ five-cores of $C_6$; on the remaining $384$ cores of $C_6$ the ground-state densities are not faithful, and on $112$ of them, as on the larger cores of $C_4$, the minimum itself exceeds $1/(k-1)$ (Sec.~\ref{sec:clusters}). It also fails on scenarios between a core and a full Mermin set (Sec.~\ref{sec:tightness}). The question the computations leave open is stated at the end of Sec.~\ref{sec:tightness}.

\subsection{Cluster-state cores}\label{sec:clusters}

Theorem~\ref{prop:core} below gives the correlator-only minimum of any core. Whether the faithful minimum is the same---whether the marginals are free, as Theorem~\ref{thm:reduction} makes them for Mermin scenarios---was tested on the minimal contradictions of three graph states not local-Clifford equivalent to a GHZ state: the linear cluster $C_4$ and the ring clusters $C_5$ and $C_6$ (conventions in Sec.~\ref{sec:scenario}). A set of signed stabilizers is a contradiction iff the $(\text{party}, \text{Pauli})$ occurrence vectors sum to zero over $\mathrm{GF}(2)$ while the signs multiply to $-1$, and the cores are the minimal dependent sets (circuits) of those vectors with sign product $-1$, since a proper dependent subset of a contradiction, or its complement, is a smaller contradiction; for $C_4$ and $C_5$ they were enumerated completely from the kernel of the occurrence map (script \texttt{core\_census\_v1.py}), for $C_6$ only up to size $5$ (\texttt{cluster\_cores\_v4.py}). None of the three has a contradiction of size $2$ or $3$. The counts are in Table~\ref{tab:cores}.

Three certificates are used. (i) \emph{Exact, construction:} lower bound Theorem~\ref{thm:lowerbound} on the core, $1/(k-1)$; upper bound the ground-state construction, with every marginal checked in integer arithmetic. (ii) \emph{Exact, pair of certificates:} an exact lower bound from the dual of the faithful LP, rationalized and verified in rational arithmetic (\texttt{core\_dual\_certificate\_v2.py}), matched by an exact faithful model: either a rationalized vertex solution of the faithful LP verified in rational arithmetic (\texttt{c6\_exact\_primal\_v2.py}), or the \emph{tightened construction}---for each setting $s$ let $B_s$ be the set of strategies satisfying every deterministic constraint of $s$ (its full product and its $\pm 1$ marginals) and $H'(\lambda)$ the number of settings $s$ with $\lambda \notin B_s$; $\rho_s$ is uniform on the minimizers of $H'$ within $B_s$, and its zero-target marginals are then checked exactly (\texttt{core\_tight\_bound\_v3.py}). A tightened Theorem~\ref{thm:lowerbound}, with each sector cut down to $B_s$, is also exact but gives only $1/(k-1)$ on the cores below; the dual certificate is needed. (iii) \emph{LP:} the faithful LP on the collapsed strategy space, double precision.

\begin{table}[t]
\caption{\label{tab:cores}Faithful minimum of $F$ on the cores of three cluster states, by size $k$. ``Faithful'' is the number of cores on which the ground-state construction reproduces every marginal and so certifies $F_{\min} = 1/(k-1)$ exactly (certificate i). Certificate ii: exact lower bound from a rational LP dual certificate, matched by an exact faithful model (tightened construction or rationalized LP solution). Certificate iii: exact---the lower bound $1/(k-1)$ of Theorem~\ref{prop:core} is matched by a rationalized LP solution verified in rational arithmetic. The $C_6$ census covers sizes $\le 5$ only; its rows are split by certificate (Sec.~\ref{sec:clusters}).}
\begin{ruledtabular}
\small\setlength{\tabcolsep}{3pt}
\begin{tabular}{llrrll}
State & $k$ & cores & constr.\ faithful & $F_{\min}$ & certificate \\
\colrule
$C_4$ chain & 4 & 8 & 8 & $1/3$ & i \\
 & 5 & 8 & 8 & $1/4$ & i \\
 & 6 & 8 & 0 & $1/4 > 1/5$ & ii \\
 & 7 & 16 & 0 & $1/4 > 1/6$ & ii \\
$C_5$ ring & 4 & 60 & 60 & $1/3$ & i \\
 & 5 & 180 & 180 & $1/4$ & i \\
 & 6 & 572 & 572 & $1/5$ & i \\
 & 7 & 1{,}980 & 1{,}980 & $1/6$ & i \\
 & 8 & 5{,}220 & 5{,}220 & $1/7$ & i \\
 & 9 & 9{,}640 & 9{,}640 & $1/8$ & i \\
 & 10 & 8{,}912 & 8{,}912 & $1/9$ & i \\
$C_6$ ring & 4 & 280 & 280 & $1/3$ & i \\
 & 4 & 80 & 0 & $1/3$ & iii \\
 & 5 & 1{,}616 & 1{,}616 & $1/4$ & i \\
 & 5 & 192 & 0 & $1/4$ & iii \\
 & 5 & 112 & 0 & $1/3 > 1/4$ & ii, all 112 \\
\end{tabular}
\end{ruledtabular}
\end{table}

$C_5$ has no core of size above $10$; on all $26{,}564$ of its cores, and on the size-4 and size-5 cores of $C_4$, the construction is faithful with $F = 1/(k-1)$. On the $24$ cores of $C_4$ of sizes $6$ and $7$ it is not, and there the faithful minimum is exactly $1/4$: the rational dual certificate gives $F_{\min} \ge 1/4$ and the tightened construction is faithful with $F = 1/4$ (both on every one of the $8$ six-cores; on every one of the $16$ seven-cores). Since $1/4 > 1/5$ and $1/4 > 1/6$, faithfulness is not free on these cores, and the faithful minimum is not a function of $k$ alone. On $C_6$ the construction is faithful on $280$ of the $360$ four-cores and $1{,}616$ of the $1{,}920$ five-cores. On the remaining $384$ the faithful LP was run on every core. All $80$ four-cores return $1/3 = 1/(k-1)$, and a rationalized vertex solution verified in rational arithmetic attains $1/3$ exactly on each of them; so every four-core of $C_6$ has faithful minimum exactly $1/3$, although the ground-state construction is an optimizer on only $280$. Of the $304$ five-cores, $192$ return $1/4 = 1/(k-1)$, exact on each by the same rationalized solution, and $112$ return $1/3$; on each of these $112$ the rational dual certificate gives $F_{\min} \ge 1/3$ and the tightened construction is faithful with $F = 1/3$, so their faithful minimum is exactly $1/3$, above $1/(k-1) = 1/4$. One of them is $+IIZYYZ$, $+IZYYZI$, $-IZYXYZ$, $+XZIIIZ$, $+XZZXZZ$.

Padding the $C_4$ four-core $+IZXZ$, $+IZYY$, $+ZYYZ$, $-ZYXY$ with two further stabilizers of its group raises $s$ to $5/6$ for each of the $55$ possible choices. For $43$ of them the ground-state construction of the padded scenario is still faithful with $F = 1/3$; since adding constraints cannot lower the minimum, $F_{\min} = 1/3$ exactly there and $F \cdot s = 5/18$. For the other $12$ choices the construction either fails a marginal constraint ($8$) or is faithful with $F = 1/2$ or $1$ ($4$); the LP returns $1/3$ for these twelve as well (exact lower bound $1/3$ from the core). So the identity $F \cdot s = 1/4$ holds for four-element cores and for the full Mermin sets (Theorem~\ref{thm:staircase}) and fails already for the padded $C_4$ scenarios ($5/18$) and for five-cores ($1/5$). Nor does a general scenario simply inherit the minimum of its tightest embedded core: in Mermin-5 with one constraint deleted (Sec.~\ref{sec:tightness}) every embedded core has minimum at most $1/3$ (Theorem~\ref{prop:core}) while the scenario's computed minimum is $2/5$.

\subsection{Cores, and the tightness of the lower bound}\label{sec:tightness}

\begin{theorem}[Cores]\label{prop:core}
Let $S'$ be a core---a minimal contradiction---of size $k$ among constraints whose functionals are $\mathrm{GF}(2)$-linear in the bits of a strategy set $\Lambda = \mathrm{GF}(2)^N$ and whose targets are deterministic; in particular, any core of stabilizer constraints with each party holding a pre-set outcome for each local Pauli. Then the correlator-only minimum of $F$ over $S'$ is $1/(k-1)$, attained by the ground-state construction ($\rho_s$ uniform on the strategies of minimum violation count within the sector of $s$). If $S'$ is embedded in a Mermin scenario, the faithful minimum is also $1/(k-1)$.
\end{theorem}

\begin{proof}
Any $k-1$ of the $k$ functionals are linearly independent: if the functionals of a nonempty proper subset $J \subsetneq S'$ summed to zero, so would those of $S' \setminus J$, and since the target sums of $J$ and of $S' \setminus J$ add to $1$, one of the two would be a contradiction smaller than $S'$. Hence the strategies satisfying all constraints but the $j$th form an affine subspace of codimension $k-1$, of size $|\Lambda|/2^{k-1}$ with $\Lambda$ the strategy set, and they violate the $j$th automatically, since no strategy satisfies all $k$. These $k$ disjoint sets are the ground set of $H$. The sector of constraint $s$ contains $(k-1)|\Lambda|/2^{k-1}$ ground states and two sectors share $(k-2)|\Lambda|/2^{k-1}$, so the construction gives $F_{ss'} = 1 - (k-2)/(k-1) = 1/(k-1)$ for every pair, and it reproduces every correlator because each $\rho_s$ is supported in its sector. Theorem~\ref{thm:lowerbound} with $m(S') = k-1$ gives the matching lower bound. In a Mermin scenario the marginals vanish by the fiber lift of Theorem~\ref{thm:reduction}.
\end{proof}

Theorem~\ref{thm:n3} is the case $k = 4$ (Mermin-3 is itself a core), and Theorem~\ref{thm:n4}'s lower bound is Theorem~\ref{prop:core} on a type-(i) core of Appendix~\ref{app:n4}, its upper bound Lemma~\ref{lem:counts} at $n = 4$.

Theorem~\ref{thm:lowerbound} is tight on every Mermin set (Theorems~\ref{thm:allnlower} and \ref{thm:staircase}); at $n = 3, 4, 5$, complete enumeration of all subsets shows that the maximum $B$ is attained on the full set at odd $n$ and on the embedded Mermin$_{n-1}$ set at even $n$ and nowhere larger. It is not tight universally. Deleting a single constraint from Mermin-5 (any of the three orbits, $|S| = 0, 2, 4$) leaves the computed (LP) minimum at $2/5$ while $B$, maximized over all $2^{15}$ subsets, falls to $4/11$; the reduced scenario has $180$ cores, of sizes $4$ and $6$, so by Theorem~\ref{prop:core} its embedded cores have minima $1/3$ and $1/5$. The ground-state construction is not optimal universally either: in the same scenario it is feasible but gives $1/2$, because only the $12$ classes that violated the deleted constraint remain ground states. The value $2/5$ is attained by the Mermin-5 construction restricted to the fifteen remaining settings---densities uniform on the ground states of the \emph{parent} scenario---which is feasible for the reduced scenario and, by Lemma~\ref{lem:counts}, has every pair at distance $2/5$; so $F_{\min} \in [4/11, 2/5]$ rigorously, with the LP placing it at $2/5$.

The five-cores of $C_4$, $C_5$ and $C_6$ separate two formulas that agree on every Mermin set and every four-core: $1/(4s) = 5/16$ against Theorem~\ref{thm:lowerbound}'s $(5-4)/4 = 1/4$. On every five-core of $C_5$, and on $1{,}808$ of the $1{,}920$ of $C_6$, the minimum is exactly $1/4$, so $F \cdot s = 1/4$ fails on five-cores ($1/5$); when Theorem~\ref{thm:lowerbound} is tight, $F \cdot s = 1/4$ holds iff $N(N-1) = 4m(N-m)$, which is the case for the full Mermin sets and four-cores and not for five-cores. Theorem~\ref{thm:lowerbound} is the exact minimum on every core for the correlators (Theorem~\ref{prop:core}) and, with the marginals, on every core of $C_5$, on the size-4 and size-5 cores of $C_4$, on every four-core of $C_6$ and on $1{,}808$ of its five-cores.

In summary: the minimum lies between the lower bound of Theorem~\ref{thm:lowerbound} and the value of the ground-state construction whenever the latter is faithful; it coincides with both on every full Mermin set, on every core for the correlators alone, and, with the marginals, on every core of $C_5$, the size-4 and size-5 cores of $C_4$, and the $1{,}896$ cores of $C_6$ on which the construction is faithful. It lies strictly above the lower bound on the larger cores of $C_4$ and on $112$ five-cores of $C_6$ (exact certificates), and, by LP value, on Mermin-5 minus one constraint. It lies strictly below the construction's value on Mermin-5 minus one constraint (exact) and on $4$ of the $55$ padded $C_4$ scenarios, where the construction is faithful at $1/2$ or $1$ against an LP value of $1/3$ (on $8$ further padded scenarios, and on $384$ cores of $C_6$, the construction is not faithful at all). Two questions are open: on which stabilizer cores the marginals are free---they are on every core of $C_5$ and every four-core of $C_6$, and not on $C_4$ at sizes $6$ and $7$, where the minimum is $1/4$ regardless of $k$, nor on $112$ of the $1{,}920$ five-cores of $C_6$, where it is $1/3$---and what the optimal densities are on scenarios that are neither a core nor a full Mermin set.

\subsection{No proximity structure on $H$}\label{sec:landscape}

The violation count $H$ is not graded by Hamming distance between classes. For odd $n$, flipping the single bit $P$ takes every ground class to a class of maximal violation (Lemma~\ref{lem:mermsum}: $|{\textstyle\sum_S} \hat{t}_S \chi_S| = R$ at every class, so $\min H$ and $\max H$ differ only in the sign of $(-1)^P$). For even $n$, no two ground classes are at Hamming distance $1$: by Lemma~\ref{lem:mermsum} the ground classes have $|\gamma| \equiv n/2 \pmod 2$ and $P$ determined by $\gamma$, so a single bit flip leaves the ground set. At $n = 4$, moreover, no ordering of the $n+1$ class bits arranges the sets $\{c : H(c) \le h\}$ as unions of subtrees of the binary tree on those bits ($0$ of $120$ orderings, by enumeration; script \texttt{ultrametric\_check.py} in the repository).

\section{Discussion}

\emph{Operational reading.} $F_{\min}(n)$ is the least measurement-dependence fraction (one half of Hall's $M$) with which a local deterministic model reproduces the maximal violation of the $n$-party Mermin inequality. Equivalently, it is the least settings dependence with which a local deterministic source can produce all $2^{n-1}$ perfect correlators, the maximal value of the Mermin expression. By Theorem~\ref{thm:staircase} that requirement is $R/[2(R+1)]$ exactly, and never exceeds $1/2$---half of the maximal dependence $F = 1$---at any $n$, a statement independent of interpretation.

\emph{Relation to exclusion results.} Vieira, Ramanathan, and Cabello~\cite{Vieira2025} show that certain quantum correlations cannot be reproduced by any hidden-variable model with $\ell$-measurement dependence for any $\ell > 0$---models in which every joint settings choice has probability at least $\ell$ for every hidden state~\cite{Putz2014}---so that if measurement independence is the assumption that fails, surviving models must assign zero probability to some settings for some hidden states. The models considered here all belong to that surviving class: with deterministic targets, Lemma~\ref{lem:support} obliges every faithful model to place zero weight on every class violating the relevant constraint, and since no class satisfies all constraints, every hidden state excludes some settings. Exclusion in the minimum-probability parameterization is thus compatible with a bounded cost in the total-variation one: the values computed here are the exact cost, in the normalized form of Hall's measure, of precisely the models that survive the exclusions of Ref.~\cite{Vieira2025}, and the ceiling theorem keeps that cost strictly below total dependence at every size.

\emph{Implications for measurement-dependent proposals.} Several research programs propose lawful (nomic, non-conspiratorial) measurement dependence as the assumption to relax: invariant-set theory \cite{Palmer}, the cellular-automaton interpretation \cite{tHooft2016}, future-input-dependent dynamics \cite{HossenfelderPalmer2020,DonadiHossenfelder2022}, and all-at-once or retrocausal formulations \cite{WhartonArgaman2020,Adlam2018}. The values computed here are a quantitative requirement common to all of them: any such mechanism must have measurement-dependence fraction at least $F_{\min}(n)$ at every $n$; and the function $H$ whose ground states carry the optimal densities of Theorem~\ref{thm:staircase} is not two-body in the class variables $(x_0, x_1, \ldots, x_n)$ (Sec.~\ref{sec:spectrum}).

\emph{Outlook.} The closed form is exact at every size (Theorem~\ref{thm:staircase}), and the correlator-only minimum of every core is $1/(k-1)$ (Theorem~\ref{prop:core}). What is not settled is when the marginals are free: they are on every Mermin scenario, on every core of $C_5$ and on every four-core of $C_6$, but not on the larger cores of $C_4$, where the faithful minimum is $1/4$ for $k = 6$ and $k = 7$ alike, nor on $112$ five-cores of $C_6$, where it is $1/3$; and the optimal densities on scenarios between a core and a full Mermin set are not characterized. A related excess appears on the complete $X/Y$ setting set at even $n$, where Wang, Yao and Chen find that the joint cost of reproducing both coherence components can exceed the cost of either alone \cite{Wang2026}; whether the two phenomena have a common description is open.

\section{Methods}\label{sec:methods}

\emph{Computation.} The linear programs of Secs.~\ref{sec:scenario}--\ref{sec:reduction} were solved in double precision with the open-source HiGHS solver (via SciPy), using two independent algorithms---dual simplex on the symmetry-reduced program and interior point on the class-level program---on independent hardware; candidate rationals were identified from the numerical optima by bounded-denominator reconstruction. Exact certification is separate and uses no floating point in any asserted quantity: \texttt{exact\_certificates.py} ($n \le 13$, by tabulating the satisfaction of every setting at every class) and \texttt{exact\_certificates\_fwht\_v2.py} ($n \le 22$; violation counts from an integer Walsh--Hadamard transform of the targets, Eq.~\eqref{eq:fourierH}, memory $O(2^n)$) evaluate the maximum satisfiability $m$, the ground set, and the sector overlaps of the ground-state construction per $S_n$-orbit, and assert at every $n$ that the bounds of Theorem~\ref{thm:lowerbound} and of the construction coincide at $R/[2(R+1)]$ and that the counts of Lemmas~\ref{lem:mermsum}--\ref{lem:counts} hold; the two scripts agree on every count where both run. This check uses the definition of $H$, Theorems~\ref{thm:symmetry}--\ref{thm:lowerbound} and enumeration, but none of the closed-form evaluations of Lemmas~\ref{lem:mermsum}--\ref{lem:counts}. The scope computations of Sec.~\ref{sec:structure} are: \texttt{pol\_scope\_checks\_v3.py} (Mermin-5 minus one constraint: LP, the exact bound $B$ over all subsets, the cores, the construction; and the $C_4$ core at strategy level), \texttt{c4\_padding\_construction\_v2.py} (the $55$ padded $C_4$ scenarios, construction and marginals in integer arithmetic), \texttt{c4\_bad\_pairs\_lp\_v2.py} (LP on the $12$ remaining), \texttt{cluster\_cores\_v4.py} (enumeration of the minimal contradictions of $C_4$, $C_5$, $C_6$ up to a given size; for each core the Theorem~\ref{thm:lowerbound} bound and the construction with every marginal checked in integer arithmetic; the collapsed-space faithful LP on cores where the construction is not faithful), \texttt{core\_census\_v1.py} (the complete census of $C_4$ and $C_5$ at all sizes, by kernel enumeration), \texttt{core\_tight\_bound\_v3.py} (the tightened construction, exact), \texttt{core\_dual\_certificate\_v2.py} (exact lower bounds from rationalized LP duals) and \texttt{c6\_above\_certify\_v3.py} (both certificates on every $C_6$ five-core whose LP value exceeds $1/4$), \texttt{c6\_exact\_primal\_v2.py} (rationalized vertex solutions, verified in rational arithmetic, on every $C_6$ core whose LP value equals $1/(k-1)$ and whose construction is not faithful) and \texttt{c6\_sweep\_tally\_v1.py} (the tally of the LP sweep). Every count and value quoted in Sec.~\ref{sec:structure} is asserted by the script that produced it, with the expected value stated in the script. The landscape scripts of Sec.~\ref{sec:landscape} (\texttt{ultrametric\_check.py}) and the $n \le 13$ certifier are part of the original repository suite.

\emph{AI-assisted preparation.} Large-language-model assistance (Claude, Anthropic) was used during this research as a tool for checking derivations, drafting text, and writing verification code, under the sole direction of the author, who takes full responsibility for the entire content. All exact values and identities were verified independently of the tool by machine computation as described above.

\begin{acknowledgments}
The author thanks M.~J.~W.~Hall for the body of work that made the question precise.

\emph{In memoriam.} Michael J.~W.~Hall, whose open problem this paper addresses, died in April 2026. This work is dedicated to his memory.
\end{acknowledgments}

\section*{Declarations}

The author declares no conflicts of interest. This work received no external funding and was conducted independently. This is a theoretical study; it involved no human participants, human data or tissue, or animals, and ethics approval is not applicable.

\section*{Data availability}

All code and data supporting this work, including the scripts named in Sec.~\ref{sec:methods}, the fiber-reduced linear program, and the symmetry-reduced program of Theorem~\ref{thm:symmetry}, are archived at \url{https://doi.org/10.5281/zenodo.23404683} (version 2 of the record; version 1, \url{https://doi.org/10.5281/zenodo.21445609}, holds the files of arXiv v1 and an extended companion version of this work) and maintained at \url{https://github.com/RandomInternetPreson/moire-phase-space-sampler/tree/main/DST_Bell_MI/Price_of_Locality_Source_Material}.

\appendix

\section{Proof of Theorem~\ref{thm:reduction} (fiber reduction)}\label{app:reduction}

Throughout, a deterministic strategy is $\lambda = (\alpha, \beta) \in \mathrm{GF}(2)^n \times \mathrm{GF}(2)^n$ with $a_i = (-1)^{\alpha_i}$ and $b_i = (-1)^{\beta_i}$; for the setting with $Y$-subset $S \subseteq [n]$, party $i$ outputs $a_i$ if $i \notin S$ and $b_i$ if $i \in S$. Write $P := \bigoplus_i \alpha_i$, $\gamma_i := \alpha_i \oplus \beta_i$, and let $c(\lambda) := (P, \gamma) \in \mathrm{GF}(2)^{n+1}$ be the class of $\lambda$.

\begin{lemma}[Product formula]\label{lem:product}
For any subset $T \subseteq [n]$ and setting $S$, the product of the outputs of the parties in $T$ is
\begin{equation}
\prod_{i \in T \setminus S} a_i \prod_{i \in T \cap S} b_i
\;=\; (-1)^{\,\sum_{i \in T} \alpha_i \,+\, \sum_{i \in T \cap S} \gamma_i}.
\label{eq:product}
\end{equation}
In particular, for $T = [n]$ the product is $(-1)^{P \oplus \bigoplus_{i \in S} \gamma_i}$, a function of the class alone.
\end{lemma}

\begin{proof}
Substitute $\beta_i = \alpha_i \oplus \gamma_i$ into $\prod_{T \setminus S} (-1)^{\alpha_i} \prod_{T \cap S} (-1)^{\beta_i}$ and collect exponents modulo 2. For $T = [n]$, $\sum_i \alpha_i \equiv P$.
\end{proof}

\begin{lemma}[Fiber structure]\label{lem:fiber}
Each class $(P, \gamma)$ is the image of exactly $2^{n-1}$ strategies: $\alpha$ ranges over the affine coset $A_P := \{\alpha : \bigoplus_i \alpha_i = P\}$ of the parity hyperplane $H_0 := \{h : \bigoplus_i h_i = 0\}$, with $\beta = \alpha \oplus \gamma$ determined. The group $H_0$, acting by $g_h : (\alpha, \beta) \mapsto (\alpha \oplus h, \beta \oplus h)$, preserves every class and acts simply transitively on each fiber.
\end{lemma}

\begin{proof}
Immediate: $g_h$ fixes $\gamma$ and fixes $P$ exactly when $\bigoplus_i h_i = 0$, and translation by $H_0$ is simply transitive on any coset $A_P$.
\end{proof}

\begin{lemma}[Marginal annihilation]\label{lem:annihilation}
Let $\rho_s$ be fiber-uniform, i.e., constant on each fiber. Then every proper-subset correlator vanishes: for all settings $S$ and all $T$ with $\emptyset \ne T \subsetneq [n]$,
$\sum_\lambda \rho_s(\lambda) \prod_{i \in T \setminus S} a_i \prod_{i \in T \cap S} b_i = 0$.
\end{lemma}

\begin{proof}
By Lemma~\ref{lem:product}, the summand within the fiber of class $(P,\gamma)$ carries the constant sign $(-1)^{\sum_{T \cap S} \gamma_i}$ times the character value $\chi_T(\alpha) := (-1)^{\sum_{i \in T} \alpha_i}$. Summing over the fiber sums $\chi_T$ over the coset $A_P = \alpha_0 \oplus H_0$:
$\sum_{\alpha \in A_P} \chi_T(\alpha) = \chi_T(\alpha_0) \sum_{h \in H_0} \chi_T(h)$,
which vanishes unless $\chi_T$ restricts trivially to $H_0$. The characters trivial on the parity hyperplane are exactly $\chi_\emptyset$ and $\chi_{[n]}$; since $T$ is proper and nonempty, the fiber sum is zero, and hence so is the full correlator sum.
\end{proof}

\begin{lemma}[Non-expansive uniformization]\label{lem:nonexpansive}
Define the fiber-uniformization $U\rho := |H_0|^{-1} \sum_{h \in H_0} \rho \circ g_h^{-1}$, applied simultaneously to all $\rho_s$. Then $U$ preserves normalization and every full correlator, produces fiber-uniform densities, and does not increase any pairwise $L_1$ distance:
$\| U\rho_s - U\rho_{s'} \|_1 \le \| \rho_s - \rho_{s'} \|_1$.
\end{lemma}

\begin{proof}
Each $g_h$ is a bijection of the strategy set, so $\rho \mapsto \rho \circ g_h^{-1}$ preserves normalization and $L_1$ distances; full correlators are class functions (Lemma~\ref{lem:product}) and classes are $g_h$-invariant (Lemma~\ref{lem:fiber}), so full correlators are preserved term by term. Fiber-uniformity of $U\rho$ follows from the transitivity in Lemma~\ref{lem:fiber}. The distance bound is the triangle inequality applied to the average:
$\| U\rho_s - U\rho_{s'} \|_1 \le |H_0|^{-1} \sum_h \| \rho_s \circ g_h^{-1} - \rho_{s'} \circ g_h^{-1} \|_1 = \| \rho_s - \rho_{s'} \|_1$.
\end{proof}

\begin{proof}[Proof of Theorem~\ref{thm:reduction}]
\emph{(Class solution $\Rightarrow$ faithful strategy solution, equal objective.)} Given class densities $\tilde{\rho}_s(c) \ge 0$ satisfying normalization and every full-correlator constraint
$\sum_c \tilde{\rho}_s(c)\, (-1)^{P \oplus \bigoplus_{i \in S}\gamma_i} = E(s)$,
define the lift $\rho_s(\lambda) := \tilde{\rho}_s(c(\lambda)) / 2^{n-1}$. It is normalized, reproduces every full correlator (class functions), and annihilates every proper marginal by Lemma~\ref{lem:annihilation}; its pairwise $L_1$ distances equal those of $\tilde\rho$ because within each fiber the integrand $|\rho_s - \rho_{s'}|$ is constant.

\emph{(Faithful strategy solution $\Rightarrow$ class solution, objective not larger.)} Given any faithful $\{\rho_s\}$, apply $U$ (Lemma~\ref{lem:nonexpansive}): the result is feasible, fiber-uniform---hence of the lifted form for the class marginals $\tilde\rho_s(c) := \sum_{\lambda \in c} U\rho_s(\lambda)$---and its objective is no larger. The marginal constraints on $\tilde\rho$ are vacuous, being automatic under the lift.

The two directions establish equality of the two optima, i.e., the theorem.
\end{proof}

\section{Proof of Theorem~\ref{thm:symmetry} (lossless symmetrization)}\label{app:symmetry}

Work at class level (Theorem~\ref{thm:reduction}): variables $\tilde\rho_s(c)$ with $c = (P, \gamma)$, one density per Mermin setting, identified with its $Y$-subset $S$. The symmetric group $S_n$ acts on parties, hence simultaneously on settings ($S \mapsto \pi S$) and classes ($\gamma \mapsto \pi\gamma$, $P$ fixed). Since the target \eqref{eq:targets} depends only on $|S|$, the constraint set is carried to itself: the image of a feasible point under
$(\pi \cdot \tilde\rho)_{\pi S}(P, \pi\gamma) := \tilde\rho_S(P, \gamma)$
is feasible, with equal objective (the objective is a maximum of $L_1$ norms, each preserved by the relabeling bijections).

The feasible set is convex (all constraints are affine) and the objective $\max_{\text{pairs}} \| \cdot \|_1$ is convex. Let $\tilde\rho^\ast$ be any optimum and set $\bar\rho := (n!)^{-1} \sum_{\pi \in S_n} \pi \cdot \tilde\rho^\ast$. By convexity $\bar\rho$ is feasible, and by convexity of the objective its value does not exceed the optimum; hence $\bar\rho$ is an optimum, and by construction it is covariant: $\bar\rho_{\pi S}(P, \pi\gamma) = \bar\rho_S(P, \gamma)$ for all $\pi$.

\emph{Completeness of the invariants.} A covariant family is determined by the values $\bar\rho_S(P, \gamma)$ on one representative $(S, \gamma)$ per simultaneous orbit. Two pairs $(S, \gamma)$ and $(S', \gamma')$ lie in the same orbit iff $|S| = |S'|$, $|\gamma \cap S| = |\gamma' \cap S'|$, and $|\gamma \cap \bar{S}| = |\gamma' \cap \bar{S}'|$: given equal counts, choose any bijection $S \to S'$ matching $\gamma$-support within $S$ and any bijection $\bar S \to \bar{S}'$ matching support within the complements, and extend to a permutation of $[n]$. Together with the invariant $P$, the quadruple $(|S|, P, |\gamma \cap S|, |\gamma \cap \bar S|)$ therefore separates orbits exactly, as claimed.

\emph{Variable count.} For each admissible even size $k = |S|$, the free values are indexed by $P \in \{0,1\}$, $0 \le |\gamma \cap S| \le k$, and $0 \le |\gamma \cap \bar S| \le n - k$, i.e., $2(k+1)(n-k+1)$ values; summing over the $O(n)$ sizes gives $O(n^3)$ variables in place of the $2^{n-1} \cdot 2^{n+1} = O(4^n)$ class-level variables. \hfill\qedsymbol

\section{Proof of Theorem~\ref{thm:lowerbound} (overlap bound)}\label{app:overlap}

Let $S'$ be a set of $N := |S'|$ settings whose targets are deterministic ($E = \pm 1$), and let $m := m(S')$ be the maximum, over classes $c$ (strategies, in a general scenario; classes, in a Mermin scenario, where the sums below run over classes), of the number of constraints in $S'$ satisfied by $c$. For normalized densities, write the overlap of a pair as
$\omega(s, s') := \sum_c \min\!\big(\tilde\rho_s(c), \tilde\rho_{s'}(c)\big)$,
so that $\| \tilde\rho_s - \tilde\rho_{s'} \|_1 = 2\big(1 - \omega(s,s')\big)$ and $F_{ss'} = 1 - \omega(s,s')$.

\begin{lemma}[Support restriction]\label{lem:support}
In any faithful model, $\tilde\rho_s$ is supported on the satisfying set $A_s := \{c : c \text{ satisfies } s\}$.
\end{lemma}

\begin{proof}
The full correlator is $\sum_c \tilde\rho_s(c)\, \sigma_s(c)$ with $\sigma_s(c) \in \{\pm 1\}$ and target $E(s) = \pm 1$. Since $|\sigma_s| = 1$ and $\tilde\rho_s$ is a probability density, the constraint $\sum_c \tilde\rho_s(c) \sigma_s(c) = E(s)$ holds iff $\tilde\rho_s(c) = 0$ whenever $\sigma_s(c) \ne E(s)$.
\end{proof}

\begin{lemma}[Pointwise Chebyshev step]\label{lem:chebyshev}
Fix a class $c$ and let $v_1 \le \cdots \le v_k$ be the nonzero values among $\{\tilde\rho_s(c)\}_{s \in S'}$. Then
\begin{equation}
\sum_{1 \le a < b \le k} \min(v_a, v_b) \;\le\; \frac{k-1}{2} \sum_{a=1}^{k} v_a .
\end{equation}
\end{lemma}

\begin{proof}
With the values in ascending order, $\min(v_a, v_b) = v_{\min(a,b)}$, so the left side equals $\sum_a (k - a)\, v_a$. The coefficient sequence $(k-a)_a$ is decreasing while $(v_a)_a$ is increasing; Chebyshev's sum inequality for oppositely ordered sequences gives
$\sum_a (k-a) v_a \le \tfrac{1}{k} \big(\textstyle\sum_a (k-a)\big) \big(\sum_a v_a\big) = \tfrac{k-1}{2} \sum_a v_a$.
\end{proof}

\begin{proof}[Proof of Theorem~\ref{thm:lowerbound}]
By Lemma~\ref{lem:support}, at any class $c$ the number $k$ of nonzero densities among $S'$ satisfies $k \le m$, since a density can be nonzero at $c$ only if $c$ satisfies its constraint. Summing Lemma~\ref{lem:chebyshev} over classes and using $k - 1 \le m - 1$ and normalization,
\begin{align}
\sum_{\text{pairs}} \omega(s, s')
&= \sum_c \sum_{\text{pairs}} \min\!\big(\tilde\rho_s(c), \tilde\rho_{s'}(c)\big) \nonumber\\
&\le \frac{m-1}{2} \sum_c \sum_{s \in S'} \tilde\rho_s(c)
= \frac{(m-1)N}{2}.
\end{align}
There are $N(N-1)/2$ pairs, so some pair has $\omega \le (m-1)/(N-1)$, i.e., $F_{ss'} \ge 1 - (m-1)/(N-1) = (N - m)/(N - 1)$. Since $F$ is the maximum of $F_{ss'}$ over pairs, $F \ge (N-m)/(N-1)$ for every faithful model, which is \eqref{eq:overlapbound}; maximizing over subsets $S'$ gives the bound $B$.
\end{proof}

\section{Proof of Theorem~\ref{thm:n4} ($n = 4$ minimum)}\label{app:n4}

This appendix gives the direct proof. Theorem~\ref{prop:core} and Lemma~\ref{lem:counts} now give the same result in two lines (Sec.~\ref{sec:tightness}); the catalog below is kept because it exhibits the structure of the ground states at the smallest even size.

Label the eight Mermin settings by their $Y$-subsets: the even subsets of $[4]$, namely $\emptyset$, the six pairs $\{ij\}$, and $[4]$ itself. At class level, the constraint of setting $S$ reads $f_S(c) = t_S$ with $f_S(c) := P \oplus \bigoplus_{i \in S} \gamma_i$ and targets $t_\emptyset = t_{[4]} = 0$, $t_{\{ij\}} = 1$ (the $\mathrm{GF}(2)$ encoding of \eqref{eq:targets}).

\begin{lemma}[Contradiction catalog]\label{lem:catalog}
Call a subset $C$ of settings a \emph{parity contradiction} if $\sum_{S \in C} f_S \equiv 0$ as a functional while $\sum_{S \in C} t_S \equiv 1$. There is no contradiction of size $\le 3$ containing distinct settings, and the following eight sets of size four are contradictions:
\begin{itemize}
\item[(i)] $\{\emptyset, \{ij\}, \{ik\}, \{jk\}\}$ for each $3$-subset $\{i,j,k\} \subset [4]$ (four sets);
\item[(ii)] $\{[4], \{ij\}, \{ik\}, \{il\}\}$ for each common element $i \in [4]$ (four sets).
\end{itemize}
Each setting belongs to exactly four of the eight.
\end{lemma}

\begin{proof}
$\sum_C f_S = |C| \cdot P \oplus \bigoplus_{i \in \triangle C} \gamma_i$, where $\triangle C$ is the symmetric difference of the $Y$-subsets; the functional vanishes iff $|C|$ is even and $\triangle C = \emptyset$. Size $2$ would force two settings with equal $Y$-subsets, impossible for distinct settings; odd sizes leave the $P$ term. For (i), the three pairs within $\{i,j,k\}$ have symmetric difference $\emptyset$ and target sum $0 + 1 + 1 + 1 = 1$. For (ii), the three pairs through $i$ have symmetric difference $[4]$, canceling the $[4]$ setting, again with target sum $1$. Membership count: $\emptyset$ lies in the four sets of type (i); $[4]$ in the four of type (ii); a pair $\{ij\}$ lies in the two type-(i) sets indexed by $\{i,j,k\}$ and $\{i,j,l\}$ and the two type-(ii) sets indexed by $i$ and $j$.
\end{proof}

\begin{lemma}[Maximum satisfiability]\label{lem:maxsat}
Every class violates at least two of the eight constraints, and some class violates exactly two; hence $m = 6$ and $s = 3/4$.
\end{lemma}

\begin{proof}
A class satisfying all constraints in a contradiction would contradict Lemma~\ref{lem:catalog}, so every contradiction contains a violated constraint, ruling out zero violations. If a class violated exactly one constraint $S^\ast$, all eight contradictions would have to contain $S^\ast$; but $S^\ast$ belongs to only four. Hence at least two violations. The class $P = 0$, $\gamma = (1,1,0,0)$ violates exactly $\{12\}$ and $\{34\}$ (direct check of the eight parities), so $m = 6$.
\end{proof}

\begin{lemma}[Ground-state classification]\label{lem:groundstates}
If a class violates exactly two constraints, the two violated settings form a complementary pair $\{S, \bar{S}\}$. Each of the four complementary pairs---$(\emptyset, [4])$, $(\{12\},\{34\})$, $(\{13\},\{24\})$, $(\{14\},\{23\})$---is the violation set of exactly two classes, giving eight ground-state classes in all.
\end{lemma}

\begin{proof}
Two violated settings must jointly cover all eight contradictions; since each covers exactly four, their contradiction sets must be disjoint. If two distinct settings are \emph{not} complementary, they share a contradiction: $\emptyset$ and a pair $\{jk\}$ share the type-(i) set on any $3$-subset containing $\{j,k\}$; $[4]$ and a pair share a type-(ii) set; two pairs sharing an element $i$ share the type-(i) set on their union. Two disjoint pairs and the pair $(\emptyset, [4])$ are precisely the complementary pairs (in $n=4$, disjoint pairs are complements), and a direct check shows the four type-(i) sets and four type-(ii) sets sort disjointly between the members of each complementary pair. Hence the violated pair is complementary. Solving the six remaining satisfied parities for each pair: for $(\emptyset, [4])$, the conditions force $P = 1$ and all $\gamma_i$ equal, giving $(1; 0000)$ and $(1; 1111)$; for $(\{12\},\{34\})$, they force $P = 0$, $\gamma_1 = \gamma_2 \ne \gamma_3 = \gamma_4$, giving $(0; 1100)$ and $(0; 0011)$; the remaining pairs follow by relabeling. That is two classes per pair.
\end{proof}

\begin{proof}[Proof of Theorem~\ref{thm:n4}]
\emph{Upper bound.} For each setting $S$, let $\tilde\rho_S$ be uniform (weight $1/6$ each) on the six ground-state classes whose violated pair does not contain $S$---equivalently, uniform on the ground states of the frustration Hamiltonian within the sector satisfying $S$. Two settings in the same complementary pair receive identical densities (overlap $1$); two settings in different pairs share the ground states of the two pairs containing neither, i.e., four common classes of weight $1/6$, for overlap $2/3$ and $F_{ss'} = 1/3$. Hence $F = 1/3$, marginals vanishing by the fiber lift of Theorem~\ref{thm:reduction}.

\emph{Lower bound.} Apply Theorem~\ref{thm:lowerbound} to the type-(i) contradiction $S' = \{\emptyset, \{12\}, \{13\}, \{23\}\}$: it is unsatisfiable in full ($m(S') \le 3$) and three of its constraints are simultaneously satisfiable (e.g., the class $P = 0$, $\gamma = (1, 0, 1, 0)$ satisfies $\emptyset$, $\{12\}$, $\{23\}$), so $m(S') = 3$ and $F \ge (4 - 3)/(4 - 1) = 1/3$.

The bounds meet at $1/3$.
\end{proof}

\section{Fourier spectrum of the frustration Hamiltonian}\label{app:fourier}

Encode a class $c = (P, \gamma)$ by the $\pm 1$ variables $x_0 := (-1)^P$ and $x_i := (-1)^{\gamma_i}$, and for each even $Y$-subset $S$ define the character
$\chi_S(c) := x_0 \prod_{i \in S} x_i$,
a multilinear monomial of degree $|S| + 1$. By Lemma~\ref{lem:product}, $\chi_S(c)$ is the value of the full product of outputs at setting $S$, so the constraint of setting $S$ reads $\chi_S(c) = \hat{t}_S$ with $\hat{t}_S := E(s) \in \{\pm 1\}$.

\emph{Expansion.} The frustration Hamiltonian counts violated constraints:
\begin{equation}
H(c) = \sum_S \frac{1 - \hat{t}_S\, \chi_S(c)}{2}
= 2^{n-2} - \frac{1}{2} \sum_S \hat{t}_S\, \chi_S(c),
\label{eq:fourierH}
\end{equation}
using that there are $2^{n-1}$ settings. The monomials $\chi_S$ are distinct (distinct variable sets), so by uniqueness of the multilinear expansion this \emph{is} the Fourier expansion of $H$ over $\{\pm 1\}^{n+1}$: the support is exactly the $2^{n-1}$ stabilizer characters, every nonzero coefficient has magnitude $1/2$, and the occurring degrees are $|S| + 1$ for even $|S|$, i.e., the odd degrees $1, 3, \ldots, 2\lfloor n/2 \rfloor + 1$ (equal to $n$ for odd $n$ and $n+1$ for even $n$).

\begin{lemma}[Closed form of the Fourier sum]\label{lem:mermsum}
Let $N = 2^{n-1}$. For a class $c = (P, \gamma)$ with $w := |\gamma|$,
\begin{equation}
\sum_S \hat{t}_S\, \chi_S(c) \;=\; (-1)^P\, 2^{n/2} \cos\!\big(\pi (n - 2w)/4\big).
\label{eq:mermsum}
\end{equation}
Consequently $\min_c H(c) = (N - R)/2$ for odd $n$, where every class has $|\sum_S \hat{t}_S \chi_S| = R$, and $\min_c H(c) = (N - 2R)/2$ for even $n$; in both cases the maximum satisfiable count is $m = N - \min H = N(R+1)/(2R)$, i.e., $s = (R+1)/(2R)$.
\end{lemma}

\begin{proof}
Write $\mathrm{i}$ for the imaginary unit. Expand $\prod_{j=1}^n (1 + \mathrm{i}\,x_j) = \sum_{T \subseteq [n]} \mathrm{i}^{|T|} \prod_{j \in T} x_j$. The real part keeps the even $|T|$, with sign $\mathrm{i}^{|T|} = +1$ for $|T| \equiv 0$ and $-1$ for $|T| \equiv 2 \pmod 4$, which is $\hat{t}_T$; so $\sum_S \hat{t}_S \prod_{j \in S} x_j = \operatorname{Re} \prod_j (1 + \mathrm{i} x_j)$. With $x_j = \pm 1$, each factor is $\sqrt{2}\, e^{\pm \mathrm{i}\pi/4}$, and $w$ of them carry the minus sign, so the product is $2^{n/2} e^{\mathrm{i}\pi(n - 2w)/4}$; multiplying by $x_0 = (-1)^P$ gives \eqref{eq:mermsum}. For odd $n$, $n - 2w$ is odd and $|\cos(\pi(n-2w)/4)| = 1/\sqrt{2}$, so the sum has modulus $2^{(n-1)/2} = R$ at every class; for even $n$ its modulus is at most $2^{n/2} = 2R$, attained when $n - 2w \equiv 0 \pmod 4$. The values of $\min H$ follow from \eqref{eq:fourierH}, and $m = N - \min H$ gives $(N + R)/2 = N(R+1)/(2R)$ with $N = R^2$ for odd $n$, and $N/2 + R = N(R+1)/(2R)$ with $N = 2R^2$ for even $n$.
\end{proof}

This is Mermin's inequality \cite{Mermin1990} in the form used in Sec.~\ref{sec:scenario}: a deterministic assignment satisfies at most the fraction $s = (R+1)/(2R)$ of the $2^{n-1}$ conditions. The lemma also identifies the ground set of $H$: for odd $n$ every $\gamma \in \mathrm{GF}(2)^n$ is a ground state at exactly one value of $P$ (the sign of $\cos(\pi(n-2w)/4)$), so there are $2^n$ ground classes; for even $n$ the ground classes are the $2^{n-1}$ pairs $(P, \gamma)$ with $n - 2|\gamma| \equiv 0 \pmod 4$ and $P$ fixed by the sign.

\begin{lemma}[Counts for the ground-state construction]\label{lem:counts}
Let $G$ be the ground set of $H$ and, for a setting with $Y$-subset $S$, let $G_S := G \cap A_S$ be the ground classes in the sector of $S$. Then for every $n \ge 3$ and every $S$, $|G_S| = R(R+1)$; and for $S \ne S'$, $|G_S \cap G_{S'}| = R(R+2)/2$, except that for even $n$ and $S' = \bar S$ (complementary $Y$-subsets) $G_S = G_{\bar S}$.
\end{lemma}

\begin{proof}
Write $\chi'_S(\gamma) := (-1)^{|\gamma \cap S|}$, so that $\chi_S(c) = (-1)^P \chi'_S(\gamma)$, and set $c(\gamma) := \cos(\pi(n - 2|\gamma|)/4)$. Two sums over $\gamma \in \mathrm{GF}(2)^n$ are needed. First, write $k := |S|$ and $e^{\mathrm{i}\pi(n - 2|\gamma|)/4} = e^{\mathrm{i}\pi n/4}\prod_j (-\mathrm{i})^{\gamma_j}$; the sum over $\gamma$ then factorizes party by party, with factor $1 - \mathrm{i}$ for $j \notin S$ and $1 + \mathrm{i}$ for $j \in S$:
\begin{align}
\sum_\gamma c(\gamma)\, \chi'_S(\gamma) &= \operatorname{Re}\big[e^{\mathrm{i}\pi n/4}(1-\mathrm{i})^{n-k}(1+\mathrm{i})^{k}\big] \nonumber\\
&= 2^{n/2} \cos(\pi k/2) = 2^{n/2}\, \hat{t}_S ,
\label{eq:csum}
\end{align}
since $(1-\mathrm{i})^{n-k}(1+\mathrm{i})^{k} = 2^{n/2} e^{\mathrm{i}\pi(2k-n)/4}$ and $\cos(\pi k/2) = \hat{t}_S$ for even $k$. Second, for a nonempty $D \subseteq [n]$, $\sum_\gamma \chi'_D(\gamma) = 0$.

\emph{Odd $n$.} By Lemma~\ref{lem:mermsum}, $c(\gamma) = \pm 1/\sqrt{2}$, and $G = \{(P,\gamma) : (-1)^P = \sigma(\gamma)\}$ with $\sigma := \sqrt{2}\, c \in \{\pm 1\}$. A ground class lies in $A_S$ iff $\hat{t}_S\, \sigma(\gamma)\, \chi'_S(\gamma) = 1$, so
\begin{align*}
|G_S| &= \sum_\gamma \tfrac{1}{2}\big[1 + \hat{t}_S\, \sigma(\gamma) \chi'_S(\gamma)\big] \\
&= 2^{n-1} + \tfrac{1}{2}\, \hat{t}_S \cdot \sqrt{2}\, 2^{n/2} \hat{t}_S = R^2 + R ,
\end{align*}
using \eqref{eq:csum} and $2^{n-1} = R^2$, $2^{(n+1)/2} = 2R$. For $S \ne S'$, multiplying the two indicators and using $\sigma^2 = 1$ and $\chi'_S \chi'_{S'} = \chi'_{S \triangle S'}$ with $S \triangle S' \ne \emptyset$,
\begin{align*}
|G_S \cap G_{S'}| &= \tfrac{1}{4}\Big[2^n + 4R + \hat{t}_S \hat{t}_{S'} \sum_\gamma \chi'_{S \triangle S'}(\gamma)\Big] \\
&= 2^{n-2} + R = \tfrac{1}{2} R(R+2).
\end{align*}

\emph{Even $n$.} Now $c(\gamma) \in \{0, \pm 1\}$, $c^3 = c$, and $G = \{(P,\gamma) : c(\gamma) \ne 0,\ (-1)^P = c(\gamma)\}$; the indicator of $G$ at $(P,\gamma)$ is $\tfrac{1}{2}\big[c(\gamma)^2 + (-1)^P c(\gamma)\big]$. Membership in $A_S$ fixes $(-1)^P = \hat{t}_S \chi'_S(\gamma)$, so
\begin{align*}
|G_S| &= \tfrac{1}{2}\sum_\gamma c(\gamma)^2 + \tfrac{1}{2}\, \hat{t}_S \sum_\gamma c(\gamma)\chi'_S(\gamma) \\
&= \tfrac{1}{2}\cdot 2^{n-1} + \tfrac{1}{2}\cdot 2^{n/2} = R^2 + R ,
\end{align*}
since $c^2 = 1$ exactly when $|\gamma| \equiv n/2 \pmod 2$, i.e., for half of the $\gamma$, and $2^{n-2} = R^2$, $2^{n/2} = 2R$. For $S \ne S'$ the two sector conditions are compatible iff $\hat{t}_S \hat{t}_{S'} \chi'_{D}(\gamma) = 1$ with $D = S \triangle S'$, so
\begin{align*}
|G_S \cap G_{S'}| &= \tfrac{1}{4}\sum_\gamma \big[1 + \hat{t}_S \hat{t}_{S'} \chi'_D\big]\big[c^2 + \hat{t}_S \chi'_S\, c\big] \\
&= \tfrac{1}{4}\Big[2^{n-1} + 4R + \hat{t}_S \hat{t}_{S'} \sum_\gamma c(\gamma)^2 \chi'_D(\gamma)\Big].
\end{align*}
Since $c(\gamma)^2 = \tfrac{1}{2}\big[1 + (-1)^{n/2}\chi'_{[n]}(\gamma)\big]$, the last sum equals $2^{n-1}(-1)^{n/2}$ if $D = [n]$ (i.e., $S' = \bar S$) and $0$ otherwise. For $S' \ne \bar S$ this gives $2^{n-3} + R = \tfrac{1}{2}R(R+2)$; for $S' = \bar S$, $\hat{t}_S \hat{t}_{\bar S} = (-1)^{k/2}(-1)^{(n-k)/2} = (-1)^{n/2}$, and the overlap is $2^{n-2} + R = |G_S|$, so $G_S = G_{\bar S}$.
\end{proof}

\emph{Degree-$3$ terms.} For every $n \ge 3$ there are $\binom{n}{2} \ge 3$ settings with $|S| = 2$, so the degree-$3$ Fourier coefficients of $H$ are nonzero. A function built from one- and two-body terms in the variables $(x_0, x_1, \ldots, x_n)$ has Fourier support at degree $\le 2$; by uniqueness of the expansion it cannot equal $H$.

\emph{Relation to the GHZ stabilizer Hamiltonian.} Let $O_S := \bigotimes_{i \notin S} X_i \bigotimes_{i \in S} Y_i$ be the Mermin operator of setting $S$, so that $\hat{t}_S O_S$ stabilizes the GHZ state. Products of two such elements yield the two-body $Z$-stabilizers: for instance, $(\hat{t}_\emptyset O_\emptyset)(\hat{t}_{\{ij\}} O_{\{ij\}}) = (+1)(-1) \cdot (X_i Y_i)(X_j Y_j) \otimes \mathbb{1} = Z_i Z_j$, using $XY = iZ$. Together with $O_\emptyset = X^{\otimes n}$, the elements $\{Z_i Z_j\}$ generate the full $2^n$-element stabilizer group of GHZ$_n$, whose joint $+1$ eigenspace is one-dimensional. Hence the operator $\sum_S \tfrac{1}{2}(1 - \hat{t}_S O_S)$ has the prepared GHZ state as its unique ground state, and $H$ is its classicalization under the substitution $X_i \mapsto a_i$, $Y_i \mapsto b_i$ of Sec.~\ref{sec:scenario}. This is the precise sense of the statement in Sec.~\ref{sec:spectrum}: the function $H$ is the classicalization of the Hamiltonian that uniquely selects the prepared state.

\end{document}